\documentclass[11pt,centertags,reqno,twoside]{amsart}
\usepackage{amsmath,latexsym, graphicx}
\usepackage[psamsfonts]{amssymb}
\usepackage{mathptmx}
\usepackage[bookmarks]{hyperref}
\usepackage[mathcal]{euscript}
\usepackage{xcolor}
\usepackage[normalem]{ulem}
\usepackage{amsfonts}
\usepackage{amssymb}
\usepackage{amsthm}
\usepackage{bm}
\usepackage{enumitem}
\usepackage[english]{babel}
\usepackage[bookmarks]{hyperref}

\numberwithin{equation}{section}

\renewcommand{\epsilon}{\varepsilon}

\newcommand{\be}{\begin{equation}}
\newcommand{\ee}{\end{equation}}

\newcommand{\C}{\mathbb{C}}

\newcommand{\R}{\mathbb{R}}

\newcommand{\T}{\mathbb{T}}

\newcommand{\Z}{\mathbb{Z}}

{\bf}{\it}
\newtheorem{theorem}{Theorem}[section]
\newtheorem{lemma}[theorem]{Lemma}
\newtheorem{corollary}[theorem]{Corollary}

\newtheorem{definition}[theorem]{Definition}

\newtheorem{remark}[theorem]{Remark}

\date{\today}
\begin{document}

\title[Threshold spectral transition for a fermion--boson... ]{Threshold spectral transition for a fermion-boson pair on the one-dimensional lattice}

\author{S.S.Ulashov, Sh.I.Khamidov}

\address[S.S. Ulashov]{Samarkand State University, 140104, Samarkand, Uzbekistan}
\email{sobirulashov19@gmail.com}

\address[Sh.I. Khamidov]{V.I. Romanovskiy Institute of Mathematics, UzAS,Tashkent, Uzbekistan}
\email{shoh.hamidov2021@gmail.com}

\begin{abstract}
We study a two-particle lattice Schr\"odinger operator describing a
fermion-boson pair on the one-dimensional lattice $\mathbb Z$ with a
zero-range on-site interaction of strength $\mu\in\mathbb R$ and mass
ratio $\gamma>0$. Using relative coordinates, the fiber decomposition with
respect to the total quasi-momentum
$k\in\mathbb T:=(-\pi,\pi]$, and a unitary phase transformation, we obtain
the real symmetric fiber Hamiltonian $H_{\mu,\gamma}(k)$ in
$\ell^2(\mathbb Z)$. Its essential spectrum is
\[
\sigma_{\mathrm{ess}}\bigl(H_{\mu,\gamma}(k)\bigr)
=
\bigl[
2(1+\gamma)-2a_\gamma(k),
\,
2(1+\gamma)+2a_\gamma(k)
\bigr],
\]
where
\(
a_\gamma(k)
=
\sqrt{1+2\gamma\cos k+\gamma^2}.
\)

For  \(a_{\gamma}(k)>0\) and  $\mu\neq0$, the operator has a unique simple discrete
eigenvalue
\[
E_\gamma(k,\mu)
=
2(1+\gamma)
+
\operatorname{sgn}(\mu)
\sqrt{\mu^2+4a_\gamma(k)^2},
\]
and the corresponding eigenfunction is explicitly determined and decays exponentially. The eigenvalue
lies below the essential spectrum for $\mu<0$ and above it for $\mu>0$.
At the exceptional fiber \(a_{\gamma}(k)=0\), which occurs precisely at
\((\gamma,k)=(1,\pi)\), the essential spectral band collapses to
\(\{4\}\), and the unique simple eigenvalue is
\(
E_{1}(\pi,\mu)=4+\mu.
\)

We characterize the threshold spectral transition at the critical coupling $\mu=0$. For $a_\gamma(k)>0,$ both edges
\(
E_{\mathrm{thr}}^{\pm}(k)
=
2(1+\gamma)\pm2a_\gamma(k)
\)
 are threshold resonances of $H_{0,\gamma}(k)$, with explicitly
determined resonant solutions
\(
u_{\mathrm{res}}^{\pm}
\in
\ell^\infty(\mathbb Z)\setminus\ell^2(\mathbb Z).
\)
As $\mu\to0^{\pm}$, the discrete eigenvalue converges to the corresponding
threshold and satisfies
\[
\left|
E_\gamma(k,\mu)-E_{\mathrm{thr}}^{\pm}(k)
\right|
=
\frac{\mu^2}{4a_{\gamma}(k)}+O(\mu^{4}),
\]
while, after normalization at the origin, the corresponding eigenfunction
converges to the associated resonant solution.
In the strong-coupling regime, the eigenvalue is asymptotic to
\(
\mu+2(1+\gamma),
\)
while the normalized eigenfunction becomes localized at the interaction
site.

We also establish a complete asymptotic classification near the
exceptional fiber. In the joint limit
\[
\mu\to0^{\pm},
\qquad
(\gamma,k)\to(1,\pi),
\]
the threshold behavior of the discrete eigenvalue and the corresponding
eigenfunction is determined by the relative scale of \(|\mu|\) and
\(a_{\gamma}(k)\). This reveals the nonuniformity of the weak-coupling threshold asymptotics and describes the transition from threshold-resonant behavior to localization at the interaction site.

\end{abstract}

\maketitle



\section{Introduction}
Lattice Schr\"odinger operators form a fundamental class of discrete models in mathematical physics. They arise naturally in the description of quantum particles moving in periodic media and provide a mathematical framework for the study of few-particle systems in optical lattices and solid-state physics. Their spectral properties depend on the
interaction between the particles, the geometry and dimension of the
lattice, the particle masses, and the total quasi-momentum. Even in the
two-particle case, this interplay may produce isolated eigenvalues,
virtual levels, and resonance phenomena at the edges of the essential
spectrum.

In the continuous setting, threshold phenomena for few-particle Schr\"odinger operators are closely connected with scattering theory, virtual levels, and the behavior of the resolvent in a neighborhood of the spectral thresholds; see, for example, \cite{Yaf74,FMerkuriev:1993}. Similar phenomena occur for lattice Schr\"odinger operators. However, in contrast to the continuous case, the lattice dispersion relation is bounded and therefore gives rise to both lower and upper edges of the essential spectrum. Consequently, discrete eigenvalues may occur both below and above the essential spectral band. The latter possibility is specific to lattice systems and is related to the existence of repulsively bound states, which have also been observed experimentally for atom pairs in optical lattices \cite{Nature}.

The spectral theory of few-particle lattice systems  was systematically studied by Mattis \cite{Mattis:1986}; see also \cite{Mogilner:1989} for related quasiparticle models in solid-state physics. Two-particle states in Hubbard-type lattice models and the formation of bound pairs were studied in \cite{ValientePetrosyan:2008}.

Two-particle lattice Schr\"odinger operators have been studied extensively with respect to the existence, number, and location of bound states, threshold effects, and the dependence of the discrete spectrum on the total quasi-momentum; see \cite{ALMM:2006,LakaevUlashov2012,BachLakaevPedra:2017,LakaevBachPedra:2020,KholmatovLakaevAlmuratov:2020}. Models with on-site and finite-range interactions, including nearest- and next-nearest-neighbor couplings, were investigated in \cite{LakaevKholmatovKhamidov2021,UlashovKhamidovLakaev2024,LakaevKhamidovAkhmadova2024,AkhmadovaAzizova2025,LakaevLatipovaAkhmadova2025}. In the fermionic case, the number and location of discrete eigenvalues for nearest-neighbor interactions were studied in \cite{LKh2009}, while two-fermion Hamiltonians with first- and second-nearest-neighbor interactions were considered in \cite{LakaevMotovilovAbdukhakimov:2023, LakaevAbdukhakimovKhasanov2026}.

We consider a fermion--boson pair on the one-dimensional lattice with a zero-range on-site interaction of strength $\mu.$
The particles are allowed to have different masses, with mass ratio $\gamma>0.$ Our approach is based on the relative-coordinate representation. After the direct-integral decomposition with respect to the total
quasi-momentum \(k\in\mathbb T:=(-\pi,\pi],\)
a unitary phase transformation reduces
the corresponding fiber Hamiltonian in $\ell^2(\mathbb Z)$ to the real symmetric
form
\[
\bigl(H_{\mu,\gamma}(k)u\bigr)(r)
=
2(1+\gamma)u(r)
-a_{\gamma}(k)\bigl[u(r+1)+u(r-1)\bigr]
+\mu\delta_{r,0}u(r),
\qquad r\in\mathbb Z,
\]
where
\(
a_{\gamma}(k)
=
\sqrt{1+2\gamma\cos k+\gamma^2}.
\)
Thus, for $a_{\gamma}(k)>0$, the fiber Hamiltonian is a constant-coefficient
Jacobi operator in $\ell^2(\mathbb Z)$ with a rank-one perturbation at the
origin.

The spectral theory of Jacobi operators is closely connected with
orthogonal-polynomial methods and provides effective tools for studying
discrete eigenvalues, resonances, and threshold asymptotics; see
\cite{Simon2005a,Simon2005b,Teschl1999}. A relative-coordinate approach
to two-particle lattice systems with finite-range interactions was
developed in \cite{Valiente2010}. Threshold virtual states for
Schr\"odinger operators were studied in \cite{Yafaev1979},
while related threshold phenomena for Jacobi operators with finite-rank perturbations
were recently investigated in \cite{LMakarov:2026}.

The two-particle model considered in the present paper coincides, in the one-dimensional setting, with the corresponding fermion-boson model studied in \cite{LakaevUlashov2012}. The emphasis of the present work is, however, different. Whereas \cite{LakaevUlashov2012} is concerned primarily with the existence and analyticity of bound states, the present paper provides an explicit description of the spectral and threshold structure of the associated fiber Hamiltonians, including the threshold resonances and the weak- and strong-coupling asymptotics.

Related asymptotic problems for two-particle lattice Schr\"odinger operators
with zero-range interactions were studied in
\cite{LakaevKholmatov2011,LakaevKholmatov2012}, where discrete eigenvalues
were analyzed near coupling-constant thresholds and under variations of the
total quasi-momentum. In contrast, at the exceptional fiber
\((\gamma,k)=(1,\pi)\), the hopping coefficient \(a_{\gamma}(k)\) vanishes
and the essential spectral band collapses to a single point, giving rise to
a nonuniform weak-coupling regime governed by the relative scale of
\(|\mu|\) and \(a_{\gamma}(k)\).

For every $k\in\mathbb T$, the essential spectrum of
$H_{\mu,\gamma}(k)$ is given by
\[
\sigma_{\mathrm{ess}}\bigl(H_{\mu,\gamma}(k)\bigr)
=
\bigl[
2(1+\gamma)-2a_{\gamma}(k),
\,
2(1+\gamma)+2a_{\gamma}(k)
\bigr].
\]
 For every $\mu\neq0$, the operator
$H_{\mu,\gamma}(k)$ has a unique simple discrete eigenvalue
\[
E_{\gamma}(k,\mu)
=
2(1+\gamma)
+
\operatorname{sgn}(\mu)
\sqrt{\mu^2+4a_{\gamma}(k)^2}.
\]
For attractive interaction, $\mu<0$, this eigenvalue lies below the
essential spectrum, whereas for repulsive interaction, $\mu>0$, it lies
above the essential spectrum. In the nondegenerate case
$a_{\gamma}(k)>0$,
the corresponding
eigenfunction is obtained explicitly and decays exponentially as
\(|r|\to\infty\).

At the exceptional fiber $a_{\gamma}(k)=0$, which occurs precisely when
$\gamma=1$ and $k=\pi$,
the essential spectral band reduces to the single point $4$, while the discrete eigenvalue is
\[
E_{1}(\pi,\mu)=4+\mu,
\]
and the corresponding eigenfunction is supported at the interaction site.

Particular attention is paid to a complete characterization of
the threshold spectral transition of the fiber Hamiltonian $H_{\mu,\gamma}(k)$
at the critical coupling $\mu=0,$ for $a_{\gamma}(k)>0.$ As
$\mu\to0^{\pm},$
the discrete eigenvalue converges to the corresponding edge of the essential spectrum, namely,
\[
E_{\mathrm{thr}}^{\pm}(k)
=
2(1+\gamma)\pm2a_{\gamma}(k).
\]
The threshold equations
\[
H_{0,\gamma}(k)u
=
E_{\mathrm{thr}}^{\pm}(k)u
\]
admit nonzero bounded solutions given explicitly by
\[
u_{\mathrm{res}}^{\pm}(r)=C(\mp 1)^{r},\qquad C\neq0,
\qquad r\in\Z.
\]
 Hence,
\[
u_{\mathrm{res}}^{\pm}
\in
\ell^{\infty}(\mathbb Z)\setminus\ell^2(\mathbb Z),
\]
and both edges of the essential spectrum are threshold resonances of $H_{0,\gamma}(k).$

The approach to the thresholds is further quantified by the asymptotic
behavior of the eigenvalue and eigenfunction. For $a_{\gamma}(k)>0,$ as $\mu\to0^{\pm},$ one has
\[
\left|
E_{\gamma}(k,\mu)-E_{\mathrm{thr}}^{\pm}(k)\right|=\frac{\mu^2}{4a_{\gamma}(k)}+O(\mu^{4}),
\]
while, after normalization at the origin, the corresponding eigenfunction
converges to the resonant solution $u_{\mathrm{res}}^{\pm}$.
In the opposite strong-coupling regime, as $|\mu|\to\infty$,
\[
E_{\gamma}(k,\mu)
=
\mu+2(1+\gamma)+O\bigl(|\mu|^{-1}\bigr).
\]
Consequently, the normalized eigenfunction satisfies
\[
u(r)\longrightarrow \delta_{r,0},
\qquad |\mu|\to\infty,
\]
and hence becomes localized at the interaction site.

A distinct asymptotic regime emerges when the weak-coupling limit is combined with the approach to the exceptional fiber.
Since the coefficient
$$
\frac{1}{4a_{\gamma}(k)}
$$
in the threshold asymptotic formula becomes singular as
\(a_{\gamma}(k)\to0\), the fixed-fiber asymptotics are not uniform near
\((\gamma,k)=(1,\pi)\). We therefore analyze the joint limit
\[
\mu\to0^{\pm},
\qquad
(\gamma,k)\to(1,\pi),
\]
and show that the threshold behavior of both the discrete eigenvalue and the corresponding eigenfunction is determined by the relative magnitude of \(|\mu|\) and \(a_{\gamma}(k)\). This yields three distinct scaling regimes near the exceptional fiber.

The paper is organized as follows. In Section~\ref{sec:model}, we introduce the two-particle fermion-boson Hamiltonian on the one-dimensional lattice. In Section~\ref{sec:reduction}, we pass to relative coordinates, derive the direct-integral decomposition with respect to the total quasi-momentum, and reduce the fiber operators by a unitary phase transformation to real symmetric discrete Schr\"odinger operators. Section~\ref{sec:discrete-spectr} is devoted to the spectral analysis of the fiber Hamiltonians. We determine the essential and discrete spectra, derive the corresponding eigenfunctions and strong-coupling asymptotics, characterize the threshold resonance and the weak-coupling threshold asymptotics, and analyze the joint threshold behavior near the exceptional fiber.

\section{A model operator for a two-particle fermion-boson system on a lattice}\label{sec:model}

\subsection{Two-particle Hamiltonian in position space}
Let $\mathbb Z^d$ be the $d$-dimensional integer lattice, and let $(\mathbb Z^d)^m$, $m\in\mathbb N$, denote the $m$-fold Cartesian product of $\mathbb Z^d$. We denote by $\ell^2((\mathbb Z^d)^m)$ the Hilbert space of all square-summable complex-valued functions on $(\mathbb Z^d)^m$.
In this paper we focus on the one-dimensional lattice $\Z$.
The free two-particle Hamiltonian corresponding to the fermion-boson subsystem, with mass ratio $\gamma>0,$ is defined as an operator on $\ell^2(\Z^2)$ by
\begin{equation*}
(\widehat h_{0,\gamma}\widehat{\psi})(x_1,x_2)
=
\sum_{|s|\le1}\widehat{\varepsilon}(s)\Bigl[
\widehat{\psi}(x_1+s,x_2)
+ \gamma\,\widehat{\psi}(x_1,x_2+s)\Bigr],
\end{equation*}
where the function $\widehat{\varepsilon}(s)$ is given by
\begin{equation*}
\widehat{\varepsilon}(s) =
\left\{\!\!\!\!
\begin{array}{rl}
2,          & \text{if } s = 0, \\
-1, & \text{if } |s| = 1, \\
0,          & \text{if } |s| > 1.
\end{array}
\right.
\end{equation*}

We introduce a zero-range on-site interaction of strength $\mu\in\R$ and define the interacting Hamiltonian by
\begin{equation*}
\widehat{h}_{\mu,\gamma} = \widehat h_{0,\gamma} + \mu \widehat{V},
\end{equation*}
where the interaction operator $\widehat {V}$ acts as
\begin{equation*}
(\widehat {V}\widehat{\psi})(x_1,x_2)
= \delta_{x_1,x_2}\,\widehat{\psi}(x_1,x_2)
\end{equation*}
and $\delta_{x_1,x_2}$ denotes the Kronecker delta.

\section{Reduction to relative coordinates and direct integral decomposition}\label{sec:reduction}
In this section, we use the translational invariance of the two-particle Hamiltonian to derive its direct-integral decomposition. This yields a family of fiber operators parametrized by the total quasi-momentum $k\in\T.$ By  unitary phase transformation, we further reduce each fiber operator to a discrete Schr\"odinger operator with real symmetric hopping coefficients.

\subsection{Diagonal Translation Invariance and Relative Coordinates}
For each $t\in\Z$, define the diagonal  translation operator
$$T_t:\ell^{2}(\Z^2)\to\ell^{2}(\Z^2)$$
by
\begin{equation*}
(T_t\widehat{\psi})(x_1,x_2)=\widehat{\psi}(x_1+t,x_2+t),  \quad (x_1,x_2)\in\Z^2.
\end{equation*}
\begin{lemma}\label{unitaryT}
The family $\{T_t\}_{t \in \Z}$ is a unitary representation of  $\Z$ on $\ell^2(\Z^2)$, and commutes with $\widehat{h}_{\mu,\gamma\,},$ that is, $$ T_t \widehat{h}_{\mu,\gamma} = \widehat{h}_{\mu,\gamma} T_t, \quad t \in \Z. $$
\end{lemma}
\begin{proof}
For any $t,s\in\mathbb Z$, we have
$$
T_tT_s=T_{t+s},\qquad T_0=I,\qquad T_t^{-1}=T_{-t}.
$$
Moreover, using the change of variables
$$
y_i=x_i+t,\qquad i=1,2,
$$
we obtain
$$
\begin{aligned}
\|T_t\widehat\psi\|_{\ell^2(\mathbb Z^2)}^2
=
\sum_{x_1,x_2\in\mathbb Z}
\bigl|\widehat\psi(x_1+t,x_2+t)\bigr|^2
=
\sum_{y_1,y_2\in\mathbb Z}
\bigl|\widehat\psi(y_1,y_2)\bigr|^2
=
\|\widehat\psi\|_{\ell^2(\mathbb Z^2)}^2.
\end{aligned}
$$
Hence $T_t$ is unitary for every $t\in\mathbb Z$, and the family
$\{T_t\}_{t\in\mathbb Z}$
forms a unitary representation of the group $\mathbb Z$.

Since the free operator $\widehat h_{0,\gamma}$ has
translation-invariant hopping coefficients, it is invariant under
simultaneous translations of both coordinates. Therefore,
$$
T_t\widehat h_{0,\gamma}
=
\widehat h_{0,\gamma}T_t.
$$
Furthermore,
$$
\delta_{x_1+t,x_2+t}
=
\delta_{x_1,x_2},
$$
and hence the interaction operator $\widehat V$ also commutes with $T_t$:
$$
T_t\widehat V=\widehat VT_t.
$$
Consequently, linearity of $T_t$ implies
$$\begin{aligned}
T_t\widehat h_{\mu,\gamma}
=
T_t\bigl(\widehat h_{0,\gamma}+\mu\widehat V\bigr)
=
\bigl(\widehat h_{0,\gamma}+\mu\widehat V\bigr)T_t
=
\widehat h_{\mu,\gamma}T_t.
\end{aligned}
$$
Thus,
$$
T_t\widehat h_{\mu,\gamma}
=
\widehat h_{\mu,\gamma}T_t,
\qquad t\in\mathbb Z.
$$
\end{proof}

We introduce new lattice coordinates
$$
R := x_2, \qquad r := x_1 - x_2 ,
$$ where $R$
is the position of the second particle on the lattice and
$r$ is a \textbf{\textit{relative coordinate}} (the difference between the particle coordinates).
This change of variables defines a unitary operator
$U_0:\ell^2(\Z^2)\to\ell^2(\Z^2)$ by
\begin{equation*}
(U_0\widehat{\psi})(R,r)=\widehat{\psi}(R+r,R).
\end{equation*}
The inverse operator is given by
\begin{equation*}
(U_0^{-1}\widehat{\varphi})(x_1,x_2)=\widehat{\varphi}(x_2, x_1-x_2).
\end{equation*}

In the $(R,r)$-coordinates, the diagonal translations reduce to
translations in the variable $R$. For every
$\widehat\varphi\in\ell^2(\mathbb Z^2)$ and $t\in\mathbb Z$,
$$
\begin{aligned}
(U_0T_tU_0^{-1}\widehat\varphi)(R,r)
=
(T_tU_0^{-1}\widehat\varphi)(R+r,R)
=
(U_0^{-1}\widehat\varphi)(R+r+t,R+t)
=
\widehat\varphi(R+t,r).
\end{aligned}
$$
Thus, in the $(R,r)$-representation, $T_t$ acts as a shift in the coordinate $R$ and leaves the relative coordinate $r$ unchanged.

\subsection{Direct-integral decomposition.}
We define the partial Fourier transform $$\mathcal{F}_{R}:\ell^2(\Z^2)\to L^2(\T;\ell^2(\Z)),\quad \T=(-\pi,\pi]$$ by
\begin{equation*}
(\mathcal{F}_{R}\widehat{\varphi})(k,r)=\frac{1}{\sqrt{2\pi}}\sum_{R\in\Z}e^{-i kR}\widehat{\varphi}(R,r),
\qquad k\in\T,
\end{equation*}
where
$k\in\mathbb T$ is the  {\it quasi-momentum} of the two--particle system.

The partial Fourier transform satisfies the following translation relation:
\begin{equation}\label{fouriertranslation}
\Bigl(\mathcal{F}_{R}(U_0T_tU_0^{-1})\widehat{\varphi}\Bigr)(k,r)=e^{ikt}(\mathcal{F}_{R}\widehat{\varphi})(k,r).
\end{equation}
for almost every $k\in\mathbb T$ and every $r\in\mathbb Z$.

\begin{lemma}\label{directintegral}
Let $U:=\mathcal F_{R}U_{0}.$ The  Hamiltonian $\widehat{h}_{\mu,\gamma\,}$ admits the direct-integral decomposition
$$
U\widehat{h}_{\mu,\gamma}U^{-1}=\int_{\T}^{\oplus} h_{\mu,\gamma}(k)\,{dk},
$$
in the Hilbert space
$
L^{2}\bigl(\mathbb T;\ell^{2}(\mathbb Z)\bigr),
$
where
$$
h_{\mu,\gamma}(k):
\ell^{2}(\Z)
\to
\ell^{2}(\Z)
$$
is given by
\begin{align}\label{hmyufiber}
(h_{\mu,\gamma}(k)w)(r)=2(1+\gamma)w(r)&-(1+\gamma e^{-i k})w(r+1)\notag\\
&-(1+\gamma e^{i k})w(r-1)+\mu\,\delta_{r 0}\,w(r).
\end{align}
\end{lemma}
\begin{proof}
Since $U_0$ and $\mathcal F_R$ are unitary, the operator
$$
U=\mathcal F_RU_0
$$
is unitary. By Lemma~\ref{unitaryT},
$$
T_t\widehat h_{\mu,\gamma}
=
\widehat h_{\mu,\gamma}T_t,
\qquad t\in\mathbb Z.
$$
In the $(R,r)$-coordinates using \eqref{fouriertranslation}, we obtain
$$
\mathcal F_{R}U_{0}T_{t}U_{0}^{-1}\mathcal F_{R}^{-1}
=
M_{e^{ikt}},
$$
where
$M_{e^{ikt}}$ denotes multiplication by $e^{ikt}$. Therefore,
$$
M_{e^{ikt}}
\bigl(U\widehat h_{\mu,\gamma}U^{-1}\bigr)
=
\bigl(U\widehat h_{\mu,\gamma}U^{-1}\bigr)
M_{e^{ikt}},
\qquad t\in\mathbb Z.
$$

Since the multiplication operators $M_{e^{ikt}}$, $t\in\mathbb Z$, generate the von Neumann algebra of
scalar multiplication operators on
$L^2(\mathbb T;\ell^2(\mathbb Z))$, the operator
$U\widehat h_{\mu,\gamma}U^{-1}$ is decomposable(see \cite{RSIV}, Theorem XIII.85). Hence there exists
a measurable family of bounded operators $h_{\mu,\gamma}(k)$ on
$\ell^2(\mathbb Z)$ such that
$$
U\widehat h_{\mu,\gamma}U^{-1}
=
\int_{\mathbb T}^{\oplus} h_{\mu,\gamma}(k)\,dk.
$$

Let $\widehat{\psi}\in\ell^2(\Z^2)$ and  $\widehat{\varphi}=U_0\widehat{\psi}.$ Thus
$$
(U_0\widehat{h}_{0,\gamma\,}U_0^{-1}\widehat{\varphi})(R,r)=\sum_{|s|\le1}\widehat{\varepsilon}(s)\,
[\widehat{\varphi}(R,r+s)+\gamma\widehat{\varphi}(R+s,r-s)].
$$
Applying $\mathcal{F}_{R}$ and changing the summation variable in the second term gives
$$
(U\widehat{h}_{0,\gamma\,}U^{-1}w)(k,r)=\sum_{|s|\le1}\widehat{\varepsilon}(s)\,[w(k,r+s)+\gamma e^{i ks}w(k,r-s)].
$$
Since $\widehat{\varepsilon}(0)=2$, $\widehat{\varepsilon}(\pm1)=-1,$ we have
\begin{equation}\label{hnolfiber}
(h_{0,\gamma\,}(k)w)(r)=2(1+\gamma)w(r)-(1+\gamma e^{-i k})w(r+1)-(1+\gamma e^{i k})w(r-1).
\end{equation}
Moreover, since $(\widehat{V}\widehat{\psi})(x_1,x_2)=\delta_{x_1, x_2}\widehat{\psi}(x_1,x_2)$ and $x_1-x_2=r$,
\begin{equation}\label{v1fiber}
(U\,\widehat{V}\,U^{-1}w)(r)=\delta_{r, 0}\,w(r).
\end{equation}
Consequently, combining the identities  \eqref{hnolfiber} and \eqref{v1fiber} we obtain, for a.e. $k$,
\begin{align*}
(h_{\mu,\gamma}(k)w)(r)=2(1+\gamma)w(r)&-(1+\gamma e^{-i k})w(r+1)\notag\\
&-(1+\gamma e^{i k})w(r-1)+\mu\,\delta_{r, 0}\,w(r),
\end{align*}
which is \eqref{hmyufiber}.
\end{proof}

\subsection{Phase transformation of the fiber operator}
In this section, we apply a unitary phase transformation to the fiber operator $h_{\mu,\gamma}(k)$. The transformation removes the complex phases from the nearest-neighbour hopping coefficients and converts them into the real symmetric coefficient.

Introduce the $k$-dependent coefficients
$$
c_{\pm}(k)=1+\gamma e^{\pm ik},
\qquad k\in\mathbb T,\quad \gamma>0,
$$
and set
\begin{equation}\label{agammak}
a_{\gamma}(k):=|c_{-}(k)|=|c_{+}(k)|=\sqrt{1+2\gamma\cos k+\gamma^2}.
\end{equation}
Then the fiber operator $h_{\mu,\gamma}(k)$ can be written as
$$
(h_{\mu,\gamma}(k)w)(r)
=
2(1+\gamma)w(r)
-
c_{-}(k)w(r+1)
-
c_{+}(k)w(r-1)
+
\mu\delta_{r, 0}w(r).
$$
Assume that $a_{\gamma}(k)>0,$ and define
$$
\vartheta_{\gamma}(k):=-\arg c_{-}(k).
$$
Since $c_{+}(k)=\overline{c_{-}(k)},$ we have
$$
e^{i\vartheta_{\gamma}(k)}c_{-}(k)
=
e^{-i\vartheta_{\gamma}(k)}c_{+}(k)
=
a_{\gamma}(k).
$$
Define
$$
G_k:\ell^2(\mathbb Z)\to\ell^2(\mathbb Z),
\qquad
(G_ku)(r):=e^{-i\vartheta_\gamma(k)r}u(r).
$$
\begin{lemma}\label{phasetransform}
Let $k\in\mathbb T$ satisfy $a_\gamma(k)>0$. Then $G_k$ is
unitary on $\ell^2(\mathbb Z)$, and the operator
$$
H_{\mu,\gamma}(k)
:=
G_kh_{\mu,\gamma}(k)G_k^{-1}
$$
acts as
\begin{equation}\label{fibersymmetric}
(H_{\mu,\gamma}(k)u)(r)
=
2(1+\gamma)u(r)
-a_\gamma(k)\bigl[u(r+1)+u(r-1)\bigr]
+\mu\delta_{r, 0}u(r).
\end{equation}
\end{lemma}
\begin{proof}
Since
$$
\bigl|e^{-i\vartheta_\gamma(k)r}\bigr|=1,
\qquad r\in\mathbb Z,
$$
the operator $G_k$ is unitary, with
$$
(G_k^{-1}u)(r)
=
e^{i\vartheta_\gamma(k)r}u(r).
$$
For $w=G_k^{-1}u$, direct substitution gives
$$
\begin{aligned}
(G_kh_{\mu,\gamma}(k)G_k^{-1}u)(r)
&=
2(1+\gamma)u(r)
-e^{i\vartheta_\gamma(k)}c_-(k)u(r+1)\\
&\quad
-e^{-i\vartheta_\gamma(k)}c_+(k)u(r-1)
+\mu\delta_{r, 0}u(r).
\end{aligned}
$$
Using
$$
e^{i\vartheta_\gamma(k)}c_-(k)
=
e^{-i\vartheta_\gamma(k)}c_+(k)
=
a_\gamma(k),
$$
we obtain  \eqref{fibersymmetric} representation of $H_{\mu,\gamma}(k)$.
\end{proof}
\begin{remark}
For $\gamma=1$ and $k\in(-\pi,\pi)$,
$$
c_\pm(k)
=
2\cos\frac{k}{2}\,e^{\pm ik/2},
\qquad
a_1(k)=2\cos\frac{k}{2},
$$
and therefore
$$
\vartheta_1(k)=\frac{k}{2},
\qquad
(G_ku)(r)=e^{-ikr/2}u(r).
$$
At $\gamma=1$ and $k=\pi$, the hopping coefficient $c_\pm(\pi)=0$. Hence
$$
h_{\mu,1}(\pi)=4I+\mu V.
$$
Thus, no phase transformation is required.
\end{remark}

\section{Spectral properties of the fermion-boson fiber operator}\label{sec:discrete-spectr}
In this section, we investigate the spectral properties of the fermion-boson fiber operator. In particular, we determine its essential spectrum and characterize its discrete eigenvalues.

By Lemmas~\ref{directintegral} and \ref{phasetransform}, the spectral
analysis of $\widehat h_{\mu,\gamma}$ reduces to the study of the
family of operators
$$
H_{\mu,\gamma}(k):\ell^2(\mathbb Z)\to\ell^2(\mathbb Z),
\qquad k\in\mathbb T,
$$
defined by
$$
H_{\mu,\gamma}(k)=H_{0,\gamma}(k)+\mu V, \qquad\mu\in \R,
$$
where
$$
(H_{0,\gamma}(k)u)(r)
=
2(1+\gamma)u(r)
-
a_{\gamma}(k)\bigl[u(r+1)+u(r-1)\bigr], \quad u\in \ell^{2}(\Z)
$$
with $a_{\gamma}(k)$ given by \eqref{agammak},
and
$$
(Vu)(r)=\delta_{r, 0}u(r), \quad u\in \ell^{2}(\Z).
$$
\subsection{Essential spectrum of the fiber operator $H_{\mu,\gamma}(k)$}
\begin{lemma}
For every $k\in\mathbb T$, the essential spectrum of $H_{\mu,\gamma}(k)$ is
$$
\sigma_{\mathrm{ess}}\bigl(H_{\mu,\gamma}(k)\bigr)=[\varepsilon_{\min}(k), \varepsilon_{\max}(k)] $$
where
$$\varepsilon_{\min}(k)=2(1+\gamma)-2a_{\gamma}(k) \quad\text{and}\quad   \varepsilon_{\max}(k)=2(1+\gamma)+2a_{\gamma}(k).$$
\end{lemma}

\begin{proof}
Under the Fourier transform
$$
(\mathcal Fu)(p)=\frac{1}{\sqrt{2\pi}}\sum_{r\in\mathbb Z}u(r)e^{-ipr},\qquad p\in\mathbb T,
$$
the  operator $H_{0,\gamma}(k)$
becomes the multiplication operator by
$$
\varepsilon_{k,\gamma}(p)=2(1+\gamma)-2a_{\gamma}(k)\cos p.
$$
That is,
$$
(\mathcal F H_{0,\gamma}(k)\mathcal F^{-1}f)(p)=\varepsilon_{k,\gamma}(p)f(p).
$$
Since
$$\varepsilon_{\min}(k)=\min\limits_{p\in\T}\varepsilon_{k,\gamma}(p), \quad  \varepsilon_{\max}(k)=\max\limits_{p\in\T}\varepsilon_{k,\gamma}(p),$$
it follows that
$$
\sigma_{\mathrm{ess}}\bigl(H_{0,\gamma}(k)\bigr)=
\bigl[2(1+\gamma)-2a_{\gamma}(k),\,2(1+\gamma)+2a_{\gamma}(k)\bigr].
$$
Moreover,
$$
(Vu)(r)=\delta_{r, 0}u(r),
$$
so $V$ is a rank-one operator and hence compact. Since
$$
H_{\mu,\gamma}(k)=H_{0,\gamma}(k)+\mu V,
$$
Weyl's theorem on the invariance of the essential spectrum under compact perturbations implies that
$$
\sigma_{\mathrm{ess}}\bigl(H_{\mu,\gamma}(k)\bigr)
=
\sigma_{\mathrm{ess}}\bigl(H_{0,\gamma}(k)\bigr).
$$
Consequently,
$$
\sigma_{\mathrm{ess}}\bigl(H_{\mu,\gamma}(k)\bigr)
=
\big[
2(1+\gamma)-2a_\gamma(k),
\,
2(1+\gamma)+2a_\gamma(k)
\big].
$$

\end{proof}

\subsection{Discrete spectrum outside the essential spectrum}
In this section, we study the discrete spectrum of the fiber operators $H_{\mu,\gamma}(k)$ associated with the two-particle Hamiltonian. We determine explicitly the unique isolated eigenvalue lying outside the essential spectrum and investigate its dependence on the interaction strength.

Recall that $$a_{\gamma}(k)=\sqrt{1+2\gamma\cos k+\gamma^2}.$$
\begin{theorem}\label{maintheorem}
Let $\gamma>0$, $k\in\mathbb{T}$ and $\mu\in\mathbb{R}\setminus\{0\}.$
Then the following assertions hold:
\begin{enumerate}
\item
If $a_{\gamma}(k)>0,$
then the operator $H_{\mu,\gamma}(k)$ has a unique simple eigenvalue
outside its essential spectrum, given by
\begin{equation}\label{eigenexplicitly}
E_{\gamma}(k,\mu)=2(1+\gamma) +\operatorname{sgn}(\mu)\sqrt{\mu^{2}+4a_{\gamma}(k)^{2}}.
\end{equation}
A corresponding eigenfunction is
\begin{equation}\label{eigenfunc}
u(r)=C\alpha^{|r|},\qquad C\neq0,\qquad r\in\Z,
\end{equation}
where
$$
\alpha=-\operatorname{sgn}(\mu)\frac{2a_{\gamma}(k)}{|\mu|+\sqrt{\mu^{2}+4a_{\gamma}(k)^{2}}},
\qquad
0<|\alpha|<1.
$$

\item
If $a_{\gamma}(k)=0,$
then necessarily
$\gamma=1,\, k=\pi,$
and $H_{\mu,1}(\pi)$ has the unique simple eigenvalue
$$
E_{1}(\pi,\mu)=4+\mu
$$
outside its essential spectrum. A corresponding eigenfunction is
$$
u(r)=C\delta_{r,0},
\qquad
C\neq 0.
$$
\end{enumerate}
In both cases, the eigenvalue lies below the essential spectrum for
$\mu<0$ and above it for $\mu>0$. Moreover the mapping
$$
\mu
\longmapsto
E_{\gamma}(k,\mu)
$$
is strictly increasing on $\mathbb{R}\setminus\{0\}$.
\end{theorem}
\begin{proof}
 We seek $E\in \R$ and $u\neq 0, u\in\ell^2(\Z)$ such that
\begin{equation}\label{fb-eig}
H_{\mu,\gamma}(k)u = Eu.
\end{equation}
For  $|r|\ge1$, the eigenvalue equation \eqref{fb-eig} has the form
\begin{equation*}\label{rnotzero}
Eu(r) = 2(1+\gamma)u(r)-a_{\gamma}(k)\bigl[u(r+1)+u(r-1)\bigr].
\end{equation*}
At $r=0$ the interaction contributes, and we obtain
\begin{equation*}\label{onsite}
Eu(0) = 2(1+\gamma)u(0)-a_{\gamma}(k)\bigl[u(1)+u(-1)\bigr]+ \mu u(0).
\end{equation*}
Therefore, the eigenvalue equation is equivalent to the system
\begin{equation}\label{fbeigensystem1}
\begin{cases}
E\,u(r) = 2(1+\gamma)u(r)-a_{\gamma}(k)\bigl[u(r+1)+u(r-1)\bigr], & |r| \geq 1, \\
E\,u(0) = 2(1+\gamma)u(0)-a_{\gamma}(k)\bigl[u(1)+u(-1)\bigr]+ \mu u(0), & r=0.
\end{cases}
\end{equation}

 {\bf (i)} Let $a_\gamma(k)>0.$
For $r\neq 0$, the eigenvalue equation in
\eqref{fbeigensystem1} is equivalent to the second-order recurrence
$$
a_\gamma(k)u(r+1)
-\bigl(2(1+\gamma)-E\bigr)u(r)
+a_\gamma(k)u(r-1)=0.
$$
The corresponding characteristic equation is
\begin{equation}\label{characequation}
a_\gamma(k)\lambda^2
-\bigl(2(1+\gamma)-E\bigr)\lambda
+a_\gamma(k)=0.
\end{equation}
Since $E$ lies outside the essential spectrum,
$$
\left|2(1+\gamma)-E\right|>2a_{\gamma}(k),
$$
the equation \eqref{characequation} has two distinct real reciprocal roots.
Let $\alpha$ be the unique root satisfying
$$0<|\alpha|<1.$$
Square summability eliminates the growing root on each half-line. Hence
$$
u(r)=
\begin{cases}
C\alpha^r, & r\geq1,\\[1mm]
D\alpha^{|r|}, & r\leq-1.
\end{cases}
$$
The equations at $r=1$ and $r=-1$ give
$$
C=D=u(0).
$$
Since $u\neq 0$,  necessarily  $u(0)\neq 0$. Therefore,
\begin{equation}\label{eigenfunctionform}
u(r)=u(0)\alpha^{|r|},\qquad 0<|\alpha|<1,
\qquad r\in\mathbb Z.
\end{equation}
In particular, every eigenfunction corresponding to an eigenvalue outside the essential spectrum is even.

The equation  at $r=0$ of system \eqref{fbeigensystem1} gives
\begin{equation}\label{eigenform1}
E=2(1+\gamma)-2a_\gamma(k)\alpha+\mu.
\end{equation}
On the other hand, the characteristic equation \eqref{characequation} with the solution $\alpha$
gives
\begin{equation}\label{eigenform2}
E=2(1+\gamma)-a_\gamma(k)\left(\alpha+\alpha^{-1}\right).
\end{equation}
Equating the right-hand sides of \eqref{eigenform1} and \eqref{eigenform2}, we obtain
\begin{equation}\label{quadraticequat}
a_\gamma(k)\alpha^2-\mu\alpha-a_\gamma(k)=0.
\end{equation}
The roots of \eqref{quadraticequat} are
$$
\alpha_{\pm}
=
\frac{\mu\pm\sqrt{\mu^2+4a_\gamma(k)^2}}
{2a_\gamma(k)}.
$$
Since
$$
\alpha_{+}\alpha_{-}=-1,
$$
exactly one of the root satisfies
$0<|\alpha|<1.$ For $\mu\neq 0$, this root is
$$
\alpha=\frac{\mu-\operatorname{sgn}(\mu)\sqrt{\mu^{2}+4a_{\gamma}(k)^{2}}}{2a_{\gamma}(k)}\\
=-\operatorname{sgn}(\mu)\frac{2a_{\gamma}(k)}{|\mu|+\sqrt{\mu^{2}+4a_{\gamma}(k)^{2}}}.
$$
Hence,
$$0<|\alpha|=\frac{2a_{\gamma}(k)}{|\mu|+\sqrt{\mu^{2}+4a_{\gamma}(k)^{2}}}<1,
\qquad
\operatorname{sgn}(\alpha)=-\operatorname{sgn}(\mu).
$$
Substituting this value of
$\alpha$ into \eqref{eigenform1}, we obtain
$$
\begin{aligned}
E_{\gamma}(k,\mu):=E=
2(1+\gamma)
+\operatorname{sgn}(\mu)\sqrt{\mu^2+4a_\gamma(k)^2}.
\end{aligned}
$$
Moreover,
$$
\operatorname{sgn}(\mu)
\bigl(E_{\gamma}(k,\mu)-2(1+\gamma)\bigr)
=
\sqrt{\mu^{2}+4a_{\gamma}(k)^{2}}
>
2a_{\gamma}(k).
$$
Hence,
the eigenvalue lies below the essential spectrum for $\mu<0$ and above it for $\mu>0.$

The uniqueness of the admissible root $\alpha$ excludes any other eigenvalue outside the essential spectrum, while \eqref{eigenfunctionform} shows that the corresponding eigenspace is one-dimensional. Hence, $E_{\gamma}(k,\mu)$ is a simple eigenvalue.

{\bf (ii)}
 Let $a_{\gamma}(k)=0$. Then necessarily $\gamma=1$ and $k=\pi$, and
$$\sigma_{\mathrm{ess}}\bigl(H_{\mu,1}(\pi)\bigr)=\{4\}.$$
In this case, the
system \eqref{fbeigensystem1} reduces to
\begin{equation}\label{fbeigensystem2}
\begin{cases}
(E-4)u(r)=0, & |r| \geq 1, \\
(E-4-\mu)u(0)=0, & r=0.
\end{cases}
\end{equation}
For $E\neq 4$, the first equation in \eqref{fbeigensystem2} implies
$$
u(r)=0,
\qquad |r|\geq 1.
$$
Since $u\neq 0$, we  have $u(0)\neq 0$, and the second equation of \eqref{fbeigensystem2} yields
$$
E_{1}(\pi,\mu)=4+\mu.
$$
The corresponding eigenfunction is
$$
u(r)=C\delta_{r,0},
\qquad C\neq 0.
$$
Moreover,
$$
\operatorname{sgn}(\mu)
\bigl(E_{1}(\pi,\mu)-4\bigr)
=
|\mu|>0.
$$
Therefore, $E_{1}(\pi,\mu)$ lies below the essential spectrum for
$\mu<0$ and above it for $\mu>0$. Moreover,
\eqref{fbeigensystem2} excludes the existence of any other eigenvalue outside
the essential spectrum and shows that the corresponding eigenspace is
one-dimensional. Hence, $E_{1}(\pi,\mu)$ is the unique simple eigenvalue of
$H_{\mu,1}(\pi)$ outside its essential spectrum.

Finally, if \(a_{\gamma}(k)>0\), then
\[
\frac{\partial E_{\gamma}(k,\mu)}{\partial\mu}
=
\frac{|\mu|}
{\sqrt{\mu^2+4a_{\gamma}(k)^2}}
>0,
\qquad
\mu\neq0,
\]
whereas for \(a_{\gamma}(k)=0\),
\[
E_1(\pi,\mu)=4+\mu.
\]
Therefore, the function
$$
\mu\longmapsto E_{\gamma}(k,\mu)
$$
is strictly increasing on $\R\setminus\{0\}$.
\end{proof}
\begin{corollary}[\bf Strong-coupling asymptotics]
Let $\gamma>0$ and $k\in\mathbb{T}$. As $|\mu|\to\infty$,
$$
E_{\gamma}(k,\mu)
=
\mu+2(1+\gamma)+O\bigl(|\mu|^{-1}\bigr),
$$
and, provided $a_{\gamma}(k)>0$,
$$
\alpha
=
-\frac{a_{\gamma}(k)}{\mu}
+O\bigl(|\mu|^{-3}\bigr)
\longrightarrow 0.
$$
Consequently, the normalized eigenfunction
$$
u(r)
=
\left(
\frac{1-\alpha^{2}}{1+\alpha^{2}}
\right)^{1/2}
\alpha^{|r|},
\qquad r\in\Z,
$$
converges in $\ell^{2}(\Z)$ to
$\delta_{r,0}.$
Thus, in the strong-coupling limit, the bound state becomes localized at
the interaction site $r=0$.
\end{corollary}
\begin{proof}
The result follows directly from
$$
\sqrt{\mu^{2}+4a_{\gamma}(k)^{2}}
=
|\mu|+O\bigl(|\mu|^{-1}\bigr)
\qquad\text{as } |\mu|\to\infty,
$$
and the explicit formulas for $E_{\gamma}(k,\mu),$ $\alpha$  and the eigenfunction given in Theorem~\ref{maintheorem}.
\end{proof}
\begin{remark}
The  case $a_{\gamma}(k)=0$, which corresponds to
$\gamma=1$ and $k=\pi$, is obtained as the limit of the formulas derived
for $a_{\gamma}(k)>0$. Let $\mu\neq 0$ be fixed. For
$k\in(-\pi,\pi)$,
$$
a_{1}(k)=2\cos\frac{k}{2}>0,
$$
and
$$
\alpha(k)=-\operatorname{sgn}(\mu)
\frac{2a_{1}(k)}{|\mu|+\sqrt{\mu^{2}+4a_{1}(k)^{2}}}
\longrightarrow 0
\qquad\text{as } k\to\pi.
$$
Consequently,
$$
E_{1}(k,\mu)=4+\operatorname{sgn}(\mu)\sqrt{\mu^{2}+4a_{1}(k)^{2}}
\longrightarrow
4+\mu=E_{1}(\pi,\mu).
$$
and for $C\neq0,$
$$
C\alpha(k)^{|r|}\longrightarrow C\delta_{0}
\quad\text{as } k\to\pi.
$$
\end{remark}

\subsection{Threshold resonances and asymptotics}
Let
$$\ell^{\infty}(\Z)=\{u:\Z\to\C:\sup_{r\in\Z}|u(r)|<\infty\}$$
be the Banach space of bounded complex-valued functions on $\Z$.

We introduce the notation
$$
E_{\mathrm{thr}}^{-}(k):=\varepsilon_{\min}(k)=2(1+\gamma)-2a_{\gamma}(k),
\qquad
E_{\mathrm{thr}}^{+}(k):=\varepsilon_{\max}(k)=2(1+\gamma)+2a_{\gamma}(k),
$$
for the lower and upper thresholds of the essential spectrum.
\begin{definition}
Let $\gamma>0,$ $k\in\T$ and $a_{\gamma}(k)>0$. We say that
$H_{\mu,\gamma}(k)$ has a resonance at a threshold
$E_{\mathrm{thr}}^{\pm}(k)$
if the equation
$$
H_{\mu,\gamma}(k)u=E_{\mathrm{thr}}^{\pm}(k)u
$$
has a nonzero solution
$$
u\in\ell^{\infty}(\Z)
\setminus
\ell^{2}(\Z).
$$
Such a solution is called a resonant solution.
\end{definition}
The following theorem characterizes the threshold resonances of $H_{\mu,\gamma}(k).$
\begin{theorem}
Let $\gamma>0$, $k\in\mathbb{T}$, and assume that
$
a_{\gamma}(k)>0.
$
Then the threshold points $E_{\mathrm{thr}}^{\pm}(k)$ are resonances of
$H_{\mu,\gamma}(k)$ if and only if $\mu=0$. The
corresponding resonant solutions are given by
\begin{equation}\label{resonancefunction}
u_{\mathrm{res}}^{\pm}(r)=C(\mp 1)^{r},\qquad C\neq0,
\qquad r\in\Z.
\end{equation}
\end{theorem}
\begin{proof}
At $E=E_{\mathrm{thr}}^{\pm}(k)$, the characteristic equation corresponding to
\eqref{fbeigensystem1} becomes
$$
a_{\gamma}(k)(\lambda\pm 1)^{2}=0.
$$
Thus, the characteristic root is $\lambda=\mp 1$ and has multiplicity two.
Consequently, on each half-line the general solution contains a term that
grows linearly in $|r|$. Since a resonance function must be bounded, the
linearly growing terms vanish. Hence, every bounded solution has the form
$$
u(r)=
\begin{cases}
C(\mp 1)^{r}, & r\geq 1,\\[1mm]
D(\mp 1)^{r}, & r\leq -1.
\end{cases}
$$
The equations at $r=1$ and $r=-1$ imply
$$
C=D=u(0).
$$
Therefore,
$$
u(r)=u(0)(\mp 1)^{r},
\qquad r\in\Z.
$$
Substituting this expression into the equation at $r=0$ in \eqref{fbeigensystem1}, we obtain
$$
\mu u(0)=0.
$$
Since every nontrivial bounded solution satisfies $u(0)\neq 0$, it follows
that
$$
\mu=0.
$$

Conversely, if $\mu=0$, then the functions
$$
u_{\mathrm{res}}^{\pm}(r)=C(\mp 1)^{r},
\qquad r\in\Z,
$$
satisfy
$$
H_{0,\gamma}(k)u_{\mathrm{res}}^{\pm}
=
E_{\mathrm{thr}}^{\pm}(k)u_{\mathrm{res}}^{\pm},
$$
and
$$
u_{\mathrm{res}}^{\pm}\in
\ell^{\infty}(\Z)
\setminus
\ell^{2}(\Z).
$$
Thus, $E_{\mathrm{thr}}^{\pm}(k)$ are resonances if and only if $\mu=0$.
\end{proof}

\begin{corollary}[\bf Threshold asymptotics]
Let $\gamma>0$, $k\in\mathbb{T}$, and $a_{\gamma}(k)>0.$
As $\mu\to 0^{\pm},$ the discrete
eigenvalue satisfies
$$
E_{\gamma}(k,\mu)
\longrightarrow
E_{\mathrm{thr}}^{\pm}(k),
$$
and
$$
\left|
E_{\gamma}(k,\mu)-E_{\mathrm{thr}}^{\pm}(k)\right|=\frac{\mu^2}{4a_{\gamma}(k)}+O(\mu^{4}).
$$
Furthermore, choosing $C=1$ in the representations of the eigenfunction
\eqref{eigenfunc} and the corresponding resonance solution \eqref{resonancefunction}, we have
$$
u(r)\longrightarrow u_{\mathrm{res}}^{\pm}(r)=(\mp 1)^{r}\quad\text{as } \mu\to0^{\pm}
$$
pointwise on $\Z.$
\end{corollary}
\begin{proof}
By explicit form of the eigenvalue \eqref{eigenexplicitly}, as $\mu\to 0^{\pm}$,
$$
E_{\gamma}(k,\mu)
\longrightarrow
2(1+\gamma)\pm 2a_{\gamma}(k)
=
E_{\mathrm{thr}}^{\pm}(k),
$$
and
\[
\begin{aligned}
\left|E_{\gamma}(k,\mu)-E_{\mathrm{thr}}^{\pm}(k)\right|
&=\sqrt{\mu^{2}+4a_{\gamma}(k)^{2}}-2a_{\gamma}(k)\\
&=2a_{\gamma}(k)\left(\sqrt{1+\frac{\mu^{2}}{4a_{\gamma}(k)^{2}}}-1\right)\\
&=\frac{\mu^2}{4a_{\gamma}(k)}+O(\mu^{4})
\end{aligned}
\]
Moreover,
$$
\alpha
=
-\operatorname{sgn}(\mu)
\frac{2a_{\gamma}(k)}
{|\mu|+\sqrt{\mu^{2}+4a_{\gamma}(k)^{2}}}
\longrightarrow
\mp 1
\qquad
\text{as }\mu\to 0^{\pm}.
$$
Thus, choosing $C=1$ in \eqref{eigenfunc} and  \eqref{resonancefunction}, we obtain
the pointwise convergence
$$
u(r)=\alpha^{|r|}
\longrightarrow
(\mp 1)^{r}
=
u_{\mathrm{res}}^{\pm}(r)
$$
\end{proof}

\subsection{Joint threshold asymptotics near the exceptional fiber}
The preceding corollary concerns the weak-coupling limit for fixed
\(a_{\gamma}(k)>0\). Since the coefficient
\[
\frac{1}{4a_{\gamma}(k)}
\]
in the leading term becomes singular as \(a_{\gamma}(k)\to0\), the
corresponding expansion is not uniform with respect to \((\gamma,k)\).
Since
\[
a_{\gamma}(k)^2
=
1+2\gamma\cos k+\gamma^2
=
(\gamma-1)^2+4\gamma\cos^2\frac{k}{2},
\]
hence
\[
a_{\gamma}(k)\to0
\qquad\text{as}\qquad
(\gamma,k)\to(1,\pi).
\]
Moreover, the width of the essential spectrum
\[
\sigma_{\mathrm{ess}}\bigl(H_{\mu,\gamma}(k)\bigr)
=
\bigl[
2(1+\gamma)-2a_{\gamma}(k),\,
2(1+\gamma)+2a_{\gamma}(k)
\bigr]
\]
is equal to \(4a_{\gamma}(k),\) and  also tends to zero.
This leads naturally to study to the joint limit
\[
\mu\to0^{\pm},
\qquad
a_{\gamma}(k)\to0,
\]
in which the asymptotic behavior depends on the relative rates at which
\(\mu\) and \(a_{\gamma}(k)\) vanish.
\begin{theorem}
Let \(a_{\gamma}(k)>0\) and consider the joint limit
\[
\mu\to0^{\pm},
\qquad
a_{\gamma}(k)\to0,
\]
and the sign in \(E_{\mathrm{thr}}^{\pm}(k)\) corresponds to the sign of \(\mu\).
Then the following assertions hold.
\begin{enumerate}
\item If
\[\frac{|\mu|}{a_{\gamma}(k)}\to0,\]
then
\[
\left|E_{\gamma}(k,\mu)-E_{\mathrm{thr}}^{\pm}(k)\right|=
\frac{\mu^2}{4a_{\gamma}(k)}+O\left(\frac{\mu^4}{a_{\gamma}(k)^3}\right),
\]
and
\[
\alpha=-\operatorname{sgn}(\mu)+\frac{\mu}{2a_{\gamma}(k)}
+O\left(\frac{|\mu|^2}{a_{\gamma}(k)^2}\right).
\]
\item If
\[
\frac{|\mu|}{a_{\gamma}(k)}
\to c,
\qquad
0<c<\infty,
\]
then
\[
\left|
E_{\gamma}(k,\mu)-E_{\mathrm{thr}}^{\pm}(k)
\right|
=
a_{\gamma}(k)
\left(
\sqrt{c^2+4}-2
\right)
+
o\bigl(a_{\gamma}(k)\bigr),
\]
and
\[
\alpha
=
-\frac{2\operatorname{sgn}(\mu)}{c+\sqrt{c^2+4}}
+
o(1).
\]
\item If
\[
\frac{|\mu|}{a_{\gamma}(k)}
\to\infty,
\]
then
\[
\left|
E_{\gamma}(k,\mu)-E_{\mathrm{thr}}^{\pm}(k)
\right|
=
|\mu|
-
2a_{\gamma}(k)
+
\frac{2a_{\gamma}(k)^2}{|\mu|}
+
O\left(
\frac{a_{\gamma}(k)^4}{|\mu|^3}
\right),
\]
and
\[
\alpha
=
-\operatorname{sgn}(\mu)
\frac{a_{\gamma}(k)}{|\mu|}
+
O\left(
\frac{a_{\gamma}(k)^3}{|\mu|^3}
\right).
\]
\end{enumerate}
\end{theorem}

\begin{proof}
From the explicit formulas for the eigenvalue and the eigenfunction
parameter, we have
\[
\left|
E_{\gamma}(k,\mu)-E_{\mathrm{thr}}^{\pm}(k)
\right|
=
\sqrt{\mu^2+4a_{\gamma}(k)^2}-2a_{\gamma}(k)
\]
and
\[
\alpha
=
-\operatorname{sgn}(\mu)
\frac{2a_{\gamma}(k)}
{|\mu|+\sqrt{\mu^2+4a_{\gamma}(k)^2}}.
\]

{\bf (i)} Suppose first that
\[
\frac{|\mu|}{a_{\gamma}(k)}\to0.
\]
Using the expansion of the square root, we obtain
\[
\begin{aligned}
\sqrt{\mu^2+4a_{\gamma}(k)^2}
&=
2a_{\gamma}(k)
\sqrt{1+\frac{\mu^2}{4a_{\gamma}(k)^2}}
\\
&=
2a_{\gamma}(k)
+
\frac{\mu^2}{4a_{\gamma}(k)}
+
O\left(
\frac{\mu^4}{a_{\gamma}(k)^3}
\right).
\end{aligned}
\]
Hence
\[
\left|
E_{\gamma}(k,\mu)-E_{\mathrm{thr}}^{\pm}(k)
\right|
=
\frac{\mu^2}{4a_{\gamma}(k)}
+
O\left(
\frac{\mu^4}{a_{\gamma}(k)^3}
\right).
\]
Furthermore,
\[
\begin{aligned}
\alpha
&=
\frac{-2\operatorname{sgn}(\mu)}
{
\frac{|\mu|}{a_{\gamma}(k)}+\sqrt{\frac{|\mu|^2}{a_{\gamma}(k)^2}+4}
}
\\
&=
-\operatorname{sgn}(\mu)
+
\frac{\mu}{2a_{\gamma}(k)}
+
O\left(
\frac{|\mu|^2}{a_{\gamma}(k)^2}
\right).
\end{aligned}
\]

{\bf (ii)} Now let
\[
\frac{|\mu|}{a_{\gamma}(k)}\to c,
\qquad
0<c<\infty.
\]
Then
\[
\begin{aligned}
\left|
E_{\gamma}(k,\mu)-E_{\mathrm{thr}}^{\pm}(k)
\right|
&=
a_{\gamma}(k)
\left(
\sqrt{\frac{\mu^2}{a_{\gamma}(k)^2}+4}-2
\right)
\\
&=
a_{\gamma}(k)\left(\sqrt{c^2+4}-2\right)
+
o(a_{\gamma}(k)),
\end{aligned}
\]
while
\[
\begin{aligned}
\alpha
&=
\frac{-2\operatorname{sgn}(\mu)}
{
|\mu|/a_{\gamma}(k)+\sqrt{\mu^2/a_{\gamma}(k)^2+4}
}
\\
&=
\frac{-2\operatorname{sgn}(\mu)}{c+\sqrt{c^2+4}}
+
o(1).
\end{aligned}
\]

{\bf (iii)} Finally, suppose that
\[
\frac{|\mu|}{a_{\gamma}(k)}\to\infty.
\]
Since \(a_{\gamma}(k)/|\mu|\to0\),
\[
\begin{aligned}
\sqrt{\mu^2+4a_{\gamma}(k)^2}
&=
|\mu|
\sqrt{1+\frac{4a_{\gamma}(k)^2}{\mu^2}}
\\
&=
|\mu|
+
\frac{2a_{\gamma}(k)^2}{|\mu|}
+
O\left(
\frac{a_{\gamma}(k)^4}{|\mu|^3}
\right).
\end{aligned}
\]
Consequently,
\[
\left|
E_{\gamma}(k,\mu)-E_{\mathrm{thr}}^{\pm}(k)
\right|
=
|\mu|
-
2a_{\gamma}(k)
+
\frac{2a_{\gamma}(k)^2}{|\mu|}
+
O\left(
\frac{a_{\gamma}(k)^4}{|\mu|^3}
\right).
\]
Similarly,
\[
\begin{aligned}
\alpha
&=
-\operatorname{sgn}(\mu)
\frac{2a_{\gamma}(k)/|\mu|}
{
1+\sqrt{1+4a_{\gamma}(k)^2/\mu^2}
}
\\
&=
-\operatorname{sgn}(\mu)
\frac{a_{\gamma}(k)}{|\mu|}
+
O\left(
\frac{a_{\gamma}(k)^3}{|\mu|^3}
\right).
\end{aligned}
\]
The proof is complete.
\end{proof}
\begin{remark}
The theorem resolves the nonuniformity of the weak-coupling threshold
asymptotics near the exceptional fiber. The parameter \(a_{\gamma}(k)\) is the nearest-neighbor hopping coefficient in the fiber Hamiltonian  and determines the width
\(4a_{\gamma}(k)\) of the essential spectral band. Hence
\(
|\mu|/ a_{\gamma}(k)
\)
measures the interaction strength  relative to the hopping scale.
Depending on its limiting behavior, the eigenfunction approaches the threshold
resonant solution, retains a nontrivial exponentially decaying form, or becomes
localized at the interaction site. Thus the joint limit
\[
\mu\to0^{\pm},
\qquad
(\gamma,k)\to(1,\pi),
\]
is path dependent and describes the transition from threshold
resonance to site localization near the exceptional fiber.
\end{remark}

\section*{ACKNOWLEDGMENTS}
The author is grateful to Professor S.N. Lakaev for valuable discussions
and helpful comments on the manuscript.
\section*{FUNDING}
This work was supported by the Foundation for Fundamental Research of the Republic of Uzbekistan (grant No. FL-9524115052).

\end{document}